\documentclass[cis]{ipart}

\RequirePackage[OT1]{fontenc}
\RequirePackage{amsthm,amsmath,amssymb}
\RequirePackage[numbers,square]{natbib} %
\RequirePackage[colorlinks]{hyperref}

\pubyear{0000}
\volume{0}
\issue{0}
\firstpage{1}
\lastpage{9}
\received{\smonth{6} \sday{18}, \syear{2014}}

\usepackage{amsfonts}

\usepackage{graphics}
\usepackage{epstopdf}

\usepackage{latexsym}
\usepackage{amsmath,amsthm}
\usepackage{dsfont}
\usepackage[dvips]{epsfig}
\usepackage{hyperref}
\usepackage{algorithm, algorithmic}
\usepackage{url}
\usepackage[all]{xy}
\usepackage{pstricks}

\def\mod{\rm{mod}}
\def\Bern{\rm{Bern}}
\newtheorem{Theorem}{Theorem}
\newtheorem{Lemma}{Lemma}

\def\bit{\bibitem}

\begin{document}

\begin{frontmatter}
\title{Generalization of Mrs. Gerber's Lemma\thanksref{T1}}
\thankstext{T1}{This work was partially funded by a grant from the University Grants Committee of the Hong Kong Special Administrative Region (Project No.\ AoE/E-02/08) and Key Laboratory of Network Coding,
Shenzhen, China (ZSDY20120619151314964).}

\begin{aug}
\author{\fnms{Fan} \snm{Cheng}\thanksref{t1}\ead[label=e1]{fcheng@inc.cuhk.edu.hk}}



\thankstext{t1}{Institute of Network Coding, The Chinese University of Hong Kong, N.T., Hong Kong. E-mail: fcheng@inc.cuhk.edu.hk}



\end{aug}

\begin{abstract}
Mrs. Gerber's Lemma (MGL) hinges on the convexity of $H(p*H^{-1}(u))$, where $H(u)$ is the binary entropy function. In this work, we prove that
$H(p*f(u))$ is convex in $u$ for every $p\in [0,1]$ provided  $H(f(u))$ is convex in $u$,
 where $f(u) : (a, b) \to [0, \frac12]$. Moreover, our result subsumes MGL and simplifies the original proof. We show that the
generalized MGL can be applied in binary broadcast channel to simplify some  discussion.
\end{abstract}


\end{frontmatter}
\maketitle

\section{Introduction}
Mrs. Gerber's Lemma (MGL) was introduced by Wyner and Ziv \cite{WynerZivGerber73}  in 1973, which was
shown  to be a binary version of the Entropy Power Inequality (EPI) by Shamai and Wyner \cite{ShamaiWynerEPI90}.
In Witsenhausen \cite{Wit74}, MGL was generalized to arbitrary binary input-output channels. In Ahlswede and Korner \cite{AK-MGL}, they introduced the concept of the \textit{gerbator} for arbitrary discrete memoryless channel to study MGL in alphabets with higher cardinality. In  Chayat and Shamai \cite{ChayatShamai-MGL}, MGL was extended to arbitrary memoryless
symmetric channels  with binary inputs and discrete or continuous outputs.  In  Jog and Anantharam \cite{JogVen2012}, they conjectured a strengthed MGL on an arbitrary abelian group and partially
proved it.  MGL is an instrumental tool to tackle the problems related to binary channels; e.g., the capacity
region of binary symmetric broadcast channel (BS-BC) in Wyner~\cite{WynerZivGerber73A}; the capacity region of BSC-BEC broadcast channel in Nair~\cite{Nair2010}.

The rest of this work is organized as follows. In Section \ref{sec:mgl}, we introduce the necessary notation and the background.
In Section \ref{sec:gmgl}, we present our main result on the generalized MGL. In Section \ref{sec:app}, we demonstrate the power of our result by simplifying the discussion in the binary broadcast channel.

\section{Mrs. Gerber's Lemma}\label{sec:mgl}
For  $x\in [0,1]$, the binary entropy function is defined as
\begin{equation}\notag
H(x) := -x\log x - (1-x)\log (1-x)\footnote{All logarithms in this work are natural.}
\end{equation}
and the inverse of $H(x)$ is defined as
\begin{equation}\notag
H^{-1}(x)\in [0,\frac{1}{2}].
\end{equation}
Then
\begin{equation} \notag
\frac{dH}{dx} = \log \frac{1-x}{x} \text{ and }   \frac{d^2H}{dx^2} = -\frac{1}{x(1-x)}.
\end{equation}
The convolution of $p$ and $x$ is denoted by
\begin{equation}\notag
p*x := p(1-x)+(1-p)x,
\end{equation}
where $p\in [0, 1]$.
\begin{Theorem}[Mrs. Gerber's Lemma]
Let $X$ be a Bernoulli random variable and let $U$ be an arbitrary random variable. If $Z\sim$ \Bern($p$) is independent of $(X,U)$ and $Y=X+Z$ $(\mod\ 2)$, then
\begin{equation}\notag
H^{-1}(H(Y|U))\geq H^{-1}(H(X|U))*p.
\end{equation}
\end{Theorem}
MGL can be equivalently proved via the following convexity lemma about the binary entropy function.
\begin{Lemma}\notag
    $H(p*H^{-1}(u))$ is convex in $u \in [0,1]$ for every $p\in [0, 1]$.
\end{Lemma}

\section{Generalization of MGL}\label{sec:gmgl}

We prove the following generalization of Mrs. Gerber's Lemma.
\begin{Theorem}\label{G-MGL}
Let $f(u): (a,b) \to  [0,\ \frac12]$ be twice differentiable. Then for every $p \in [0, 1]$, the function $H(p*f(u))$ is convex in $u$
provided $H(f(u))$ is convex in $u$.
\end{Theorem}

\begin{proof}
The function $H(p*f(u))$ is symmetric in $p$ about $\frac{1}{2}$, hence we can assume that $p\in [0,\frac12]$.  Since $f(u) \leq \frac12$, $p*f(u)=(1-2p)f(u) + p\leq \frac12$.

The second derivative of the given expression with respect to $u$ is given~by
\begin{equation}\label{eqn:derivative}
-\frac{((1-2p)f'(u))^2}{(1-p*f(u))(p*f(u))} + (1-2p)f''(u)\log\frac{1-p*f(u)}{p*f(u)}.
\end{equation}
 The convexity of $H(f(u))$ ($p =0$ in \eqref{eqn:derivative}) implies that $f''(u)\geq 0$.

To show the convexity it suffices to show that
\begin{equation}\notag
g(p):= -(1-2p)f'(u)^2  + (1-p*f(u))(p*f(u))f''(u) \log \frac{1-p*f(u)}{p*f(u)} \geq 0.
\end{equation}
Further we know that at both $p = 0$ and $p = \frac12$ the above expression is non-negative (at $p = 0$ from
assumption).

We will show that $g(p)$ is concave in $p$ when $p\in [0, \frac12]$. Note that the function $g_1(x) = x(1-x)\log \frac{1-x}{x}$ satisfies
$$g_1'(x) = (1-2x)\log\frac{1-x}{x} - 1,\text{ and }g_1''(x) = -2\log\frac{1-x}{x} -\frac{1-2x}{x(1-x)}.$$
Thus $g_1(x)$ is concave when $x\in [0,\frac12]$, implying $g(p)$ is concave in $p$ as desired.

\end{proof}
 When $p = 0$, $H(p*f(u))=H(f(u))$. Theorem \ref{G-MGL} shows that the convexity of $H(p*f(u))$ directly follows its convexity at the endpoint $p = 0$.
MGL follows  from Theorem \ref{G-MGL} obviously, because $H(H^{-1}(u))=u$.  Also, our argument simplifies the proof of MGL in \cite{WynerZivGerber73}.

Note that
\begin{align}
&    \frac{df^{-1}}{du} = \frac{1}{f'(f^{-1}(u))}, \text{ and }   \frac{d^2f^{-1}}{du^2} =    -\frac{f''(f^{-1}(u))}{[f'(f^{-1}(u))]^3}.\notag
\end{align}
When $f(u)$ is replaced by $f^{-1}(u)$ in Theorem~\ref{G-MGL}, $H(f^{-1}(u))$ is convex in $u$ if and only if
\begin{equation}\label{eqn:inver-derivative}
     -f^{-1}(u) (1-f^{-1}(u))\frac{f''(f^{-1}(u))}{f'(f^{-1}(u))}\log \frac{1-f^{-1}(u)}{f^{-1}(u)}     \geq 1.
\end{equation}

Theorem \ref{G-MGL} relies on the twice differentiability of $f(u)$. In the next theorem, we prove a strengthened version without this constraint.
\begin{Theorem}\label{G-MGL-2}
For every $p \in [0, 1]$, the function $H(p*f(u))$ is convex in $u$
provided $H(f(u))$ is convex in $u$, where $f(u): (a,b) \to  [0,\ \frac12].$
\end{Theorem}
Though $f(u)$ is not twice differentiable, $f(u)$ is still convex by the convexity of $H(f(u))$.  Since $f''(u)$ may not exist, we need an alternative method to deal with the convexity. Next, we state some instrumental results on convex function in Pollard \cite{Pollard-MTP} (Appendix C).

A convex function is always continuous and its one-sided derivatives  always exist. For a convex function $f(x)$, denote its left-hand and right-hand derivatives by $f'_-(x)$ and $f'_+(x)$, respectively.
Furthermore, both $f'_-(x)$ and $f'_+(x)$ are increasing; i.e., 
\begin{equation}
f'_-(x_0) \leq f'_-(x_1) \text{ and } f'_+(x_0) \leq f'_+(x_1) \text{ for each } x_0 < x_1.
\end{equation}
 Conversely, when $f'_+(x)$  is increasing, $f(x)$ is convex.
\begin{Lemma}\label{convexity}
If a real-valued function $f$ has an increasing, real-valued right-hand derivative
at each point of an open interval, then $f$ is convex on that interval. 
\end{Lemma}
Now, we prove Theorem~\ref{G-MGL-2}.
\begin{proof}
As in Theorem \ref{G-MGL}, we can still assume $p\leq \frac12$. Hence $p*f(u)\leq\frac12$. 

Since $f(u)$ is convex in $u$, for each $u_0<u_1$,
\begin{equation}\notag
f'_+(u_0)\leq f'_+(u_1).
\end{equation}
Let $$s(u) := H(p*f(u)).$$
Then $s(u)$ is continuous in $u$.
Since $H(x)$ is  differentiable,
\begin{equation}\notag
s'_+(u) = (1-2p)f'_+(u)\log\frac{1-p*f(u)}{p*f(u)}.
\end{equation}
To show $s'_+(u)$ is increasing in an interval, it is equivalent to show that $s'_+(u)$ is increasing locally; i.e.,
\begin{equation} \label{obj:1}
s'_+(u_0) \leq s'_+(u_1),
\end{equation}
where $u_1>u_0$ and $u_1\to u_0$.

Since the right-hand derivative of $\log\frac{1-p*f(u)}{p*f(u)}$ exists, 
\begin{align}
& s'_+(u_1) = (1-2p)f'_+(u_1)\log\frac{1-p*f(u_1)}{p*f(u_1)} \notag\\
& = (1-2p)f'_+(u_1)\left(\log\frac{1-p*f(u_0)}{p*f(u_0)} -\frac{(1-2p)f'_+(u_0)(u_1-u_0)}{(p*f(u_0))(1-p*f(u_0))}\right).   \notag
\end{align}
To show \eqref{obj:1}, it is equivalent to show
\begin{equation}\notag
(f'_+(u_1)-f'_+(u_0))\log\frac{1-p*f(u_0)}{p*f(u_0)}-\frac{(1-2p)f'_+(u_0)f'_+(u_1)(u_1-u_0)}{(p*f(u_0))(1-p*f(u_0))}\geq 0.
\end{equation}
That is 
\begin{align}
g_2(p) :=& (f'_+(u_1)-f'_+(u_0))(p*f(u_0))(1-p*f(u_0))\log\frac{1-p*f(u_0)}{p*f(u_0)}\notag \\
         &-(1-2p)f'_+(u_0)f'_+(u_1)(u_1-u_0) \geq 0.\notag
\end{align}
Since $f'_+(u_1)\geq f'_+(u_0)$ and $p*f(u_0)\leq \frac12$, $g_2(p)$ is also concave in $p$, similar to $g(p)$. Thus, the convexity of $H(p*f(u))$ follows from the convexity at the endpoints $p=0$ and $p=\frac12$, which completes the proof.
\end{proof}
It is easy to see that Theorem \ref{G-MGL} and Theorem \ref{G-MGL-2} still hold when $p$ $\in$ $[p_0, 1-p_0]$, as long as $H(p_0*f(u))$ is convex in $u$.

\section{Application}\label{sec:app}

As another example, we give a simple proof to the following result.

\begin{Theorem}[Claim 1 in \cite{NairWang2009}]
    When $f = H(\frac{u}{2}) + H(\frac{1-u}{2})$, $H(p*f^{-1}(u))$ is
    convex in $u\in [f(0.06), f(0.5)]$ for every $p~\in [0, \frac{1}{2}]$.
\end{Theorem}
\begin{proof}
Let $t=f^{-1}(u)$, $t\in [0.06, 0.5]$. Then
\begin{align*}
 &   f'(u)= \frac{1}{2}\log \frac{1-\frac{u}{2}}{\frac{u}{2}} -\frac{1}{2}\log\frac{1-\frac{1-u}{2}}{\frac{1-u}{2}},\\
 &   f''(u) = - \frac{1}{u(2-u)} - \frac{1}{(1-u)(1+u)}.
\end{align*}
By Theorem \ref{G-MGL}, it suffices to prove that $H(f^{-1}(u))$ is convex in $u$. By \eqref{eqn:inver-derivative}, we obtain that
\begin{align}
    \frac{\frac{1}{2}\log\frac{(2-t)(1-t)}{t(t+1)}}{\log \frac{1-t}{t}} \leq -\frac{2t^2-2t-1}{(2-t)(1-t)}.\notag
\end{align}
By some algebra,
\begin{align*}
 &   \frac{\log \frac{2-t}{t+1}}{\log \frac{1-t}{t}}\leq \frac{7t-5t^2}{(2-t)(1-t)}.
 \end{align*}
That is
\begin{align*}
 &   (1-t)(2-t)\log\frac{2-t}{t+1} \leq (7t-5t^2)\log\frac{1-t}{t}.
\end{align*}
Let
    $$l(t) = (1-t)(2-t)\log\frac{2-t}{t+1}$$
and
    $$r(t)=(7t-5t^2)\log\frac{1-t}{t}.$$
The curves of the LHS ($l(t)$) and RHS $(r(t))$ are depicted in Fig.~\ref{fig:example1}. By some algebra, we have
\begin{equation}\notag
    \frac{d^2l(t)}{dt^2} = 2\log\frac{2-t}{1+t}+\frac{3(3-2t)}{(1+t)(2-t)} + \frac{6}{(1+t)^2}\geq 0
\end{equation}
and
\begin{equation}\notag
    \frac{d^2r(t)}{dt^2} = -10\log\frac{1-t}{t}-\frac{7-10t}{(1-t)t}-\frac{2}{(1-t)^2}\leq 0.
\end{equation}
When $t=0.06$,
\begin{equation}\notag
    l(t) =  1.5902  \leq r(t)=1.5958,
\end{equation}
which completes the proof.

\begin{figure}
  \includegraphics[width=250pt]{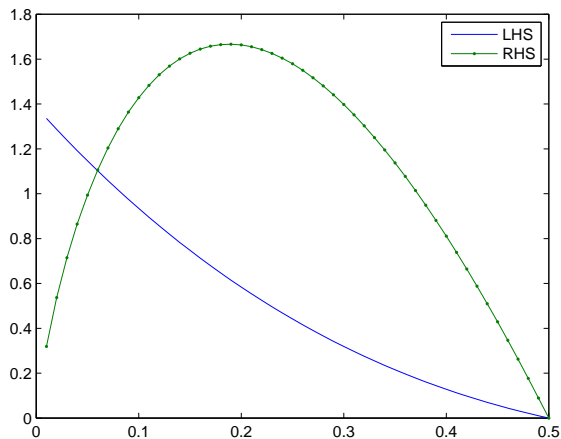}\\
  \caption{Convexity of $f(u)=H(\frac{u}{2})+H(\frac{1-u}{2})$.}\label{fig:example1}
\end{figure}

\end{proof}


\end{document}